\documentclass[runningheads]{llncs}
\usepackage[T1]{fontenc}
\usepackage{graphicx}

\usepackage{amsmath}

\usepackage{amssymb}

\usepackage{hyperref}
\let\svthefootnote\thefootnote
\newcommand\freefootnote[1]{%
  \let\thefootnote\relax%
  \footnotetext{#1}%
  \let\thefootnote\svthefootnote%
}

\usepackage[none]{hyphenat}
\usepackage[sectionbib,numbers]{natbib}
\usepackage{neuralnetwork}
\usepackage{tikz}

\usepackage[noend]{algpseudocode}
\usepackage{algorithm}

\usepackage{hhline}
\usepackage{multirow}
\usepackage{graphics}

\usepackage[misc,geometry]{ifsym}%corresponding

\usepackage{ textcomp }
\usetikzlibrary{matrix,chains,positioning,decorations.pathreplacing,arrows}
\DeclareMathOperator*{\argmax}{arg\,max}
\begin{document}

\title{Comparing Neural Network Encodings for Logic-based Explainability}

\titlerunning{Comparing Neural Network Encodings for Logic-based Explainability}

%REMOVER os % NA VERSAO FINAL
%\author{Levi Cordeiro Carvalho\inst{1} \and
%Thiago Alves Rocha\inst{2}}
%
\author{Levi Cordeiro Carvalho \and Saulo A. F. Oliveira \and
Thiago Alves Rocha \Letter}

%\orcidID{0000-1111-2222-3333}
%\orcidID{1111-2222-3333-4444}

%REMOVER os % NA VERSAO FINAL
%\authorrunning{L. Carvalho and T. Rocha}
\authorrunning{L. Carvalho, S. Oliveira and T. Rocha}


%REMOVER OS % DA VERSAO FINAL
% \institute{Federal Institute of Ceará (IFCE), Campus Fortaleza, Brazil\\
% \email{levi.carvalho@ifce.edu.br}\\
% \and
% Federal Institute of Ceará (IFCE), Campus Maracanaú, Brazil\\
% \email{thiago.alves@ifce.edu.br}}
\institute{Instituto Federal do Ceará (IFCE), Brazil\\
\email{levi.carvalho@ifce.edu.br}\\
\email{saulo.oliveira@ifce.edu.br}\\
\email{thiago.alves@ifce.edu.br} \Letter}

\maketitle              % typeset the header of the contribution
\begin{abstract}
Providing explanations for the outputs of artificial neural networks (ANNs) is crucial in many contexts, such as critical systems, data protection laws and handling adversarial examples. Logic-based methods can offer explanations with correctness guarantees, but face scalability challenges. Due to these issues, it is necessary to compare different encodings of ANNs into logical constraints, which are used in logic-based explainability. This work compares two encodings of ANNs: one has been used in the literature to provide explanations, while the other will be adapted for our context of explainability. Additionally, the second encoding uses fewer variables and constraints, thus, potentially enhancing efficiency. Experiments showed similar running times for computing explanations, but the adapted encoding performed up to 18\% better in building logical constraints and up to 16\% better in overall time.

\end{abstract}
\freefootnote{This preprint has not undergone peer review or any post-submission improvements or corrections. The Version of Record of this contribution is 
published in Intelligent Systems, LNCS, vol 15412 and is available online at \url{https://doi.org/10.1007/978-3-031-79029-4_20}.}
\section{Introduction}\label{introduction}

Artificial neural networks (ANNs) are widely applied in tasks like computer vision, speech recognition, and pattern recognition \cite{applications}. Despite their success, ANNs are often considered black-box algorithms. Such a lack of interpretability poses risks in critical domains such as medical and financial applications, where understanding model decisions is crucial. Additionally, the presence of adversarial examples highlights the need for explainability in machine learning algorithms, including neural networks. An adversarial example is an instance misclassified by a machine learning model and also slightly different from another correctly classified instance \cite{adversarial_examples}.

In this work, an explanation for a prediction made by an ANN is a subset of features and their values that alone suffice for the prediction. If an instance has the features in this subset, the ANN makes the same prediction,
regardless of the values of other features. For example, given an instance $\{sneeze=True$, $weight=70 \ kg$, $headache=True$, $age=40 \ years\}$ and its ANN output \textit{flu}, a possible explanation could be $\{sneeze=True, headache=True\}$. That is, if an instance has the features $sneeze=True$ and $headache=True$, the ANN prediction is \textit{flu}, regardless of $weight$ and $age$ values. A minimal explanation avoids redundancy by including only essential information. An explanation is considered minimal when removing any feature results in the loss of assurance that every instance satisfying the explanation maintains the same output. Then, a minimal explanation avoids redundancy, providing only essential information.

Heuristic methods, such as ANCHOR \cite{ANCHOR} and LIME \cite{LIME}, have been used to provide explanations for machine learning models. However, these approaches explore the instance space locally, not resulting in explanations that have minimal sizes and formal guarantees of correctness. Correctness guarantees are provided when there are no instances with the values specified in the explanation such that the ANN makes a different prediction. Moreover, minimal explanations are desired since they do not contain redundancy, making them easier to understand and interpret.

%Ignatiev_why
%joao,
Some approaches aim to provide explanations for machine learning models with formal guarantees of correctness \cite{shi,Ignatiev_abduction,choi,gorji2022sufficient,audemard23a,wuVerix23}. \citet{Ignatiev_abduction} proposed a logic-based algorithm that gives minimal and correct explanations for ANNs, utilizing logical constraints originally designed for finding adversarial examples \citet{fischetti}. These constraints include linear equations, inequalities, and logical implications, solved using a Mixed Integer Linear Programming (MILP) solver. However, scalability issues arise, particularly with large ANNs, necessitating further development before deployment in large-scale production environments.

This work explores two different encodings to improve the scalability of providing correct minimal explanations for ANNs, building upon \cite{Ignatiev_abduction}. In addition to the logical constraints of \cite{fischetti}, we adopt the encoding proposed by \citet{tjeng}, which uses fewer variables and constraints, and excludes logical implications. By reducing variables and constraints compared to \cite{fischetti}, our approach aims to enhance explanation computation performance. To adapt the approach of \cite{tjeng} for explanations, we introduce new constraints to ensure correctness. In line with the encodings proposed by \citet{fischetti} and \citet{tjeng}, we also compute lower and upper bounds for each neuron. These bounds are found through optimization using a MILP solver. Moreover, these bounds can aid the solver in computing explanations more rapidly. In this manner, we compare the time required for constructing logical constraints with lower and upper bounds of each neuron, along with the time needed for computing explanations.

We conducted experiments to evaluate both encodings. Our adaptation of the encoding proposed in \cite{tjeng} exhibits a better running time in building encodings for ANNs with two layers and tens of neurons, showing an improvement of up to 18\%. Surprisingly, both methods exhibit similar running times for computing explanations. Furthermore, our adaptation outperforms the other encoding in the overall time, encompassing both building logical constraints and computing explanations. In this case, the results indicate an improvement of up to 16\%. In summary, our main contributions are described in the following:

\begin{itemize}
    \item Adaptation of the encoding proposed in \cite{tjeng} to provide explanations for ANNs; additional constraints were incorporated to address the problem of computing explanations.
    
    %an encoding to verify robustness of ANNs to an encoding without indicator constraints to provide explanations for ANNs;
    %\item Publicly available implementations of both encodings for finding explanations for ANNs\footnote{\url{https://github.com/LeviCC8/Explications-ANNs}};
    \item Comparative analysis of the running time for building the logical constraints between the two approaches. Additionally, we analyze the time for generating explanations using both encodings.

    \item Publicly available implementations of both encodings for finding explanations for ANNs\footnote{\url{https://github.com/LeviCC8/Explications-ANNs}}.
    
\end{itemize}

In the next section, we review some concepts and terminologies about Logic, MILP and ANNs. Sections \ref{explanations} and \ref{encoding} show how to compute explanations with and without implications, respectively. Section \ref{encoding} describes our adaptation of the encoding proposed in \cite{tjeng}. Experiments and results are presented in Section \ref{experiments}. Finally, conclusions and future work are described in Section \ref{conclusions}.

\section{Background}\label{background}
In this section, we introduce some initial concepts and terminology to understand the rest of this work.

%We assume some familiarity with first-order logic (FOL) \cite{kroening2016decision}. 
%quantifier -free first order logic over the theory of linear real arithmetic
\subsection{First-order Logic over LRA}\label{subsec:logic}
In this work, we use first-order logic (FOL) to give explanations with guarantees of correctness. We use quantifier-free first-order formulas over the theory of linear real arithmetic (LRA). Then, first-order variables are allowed to take values from the real numbers $\mathbb{R}$. For details, see \cite{kroening2016decision}. Therefore, we consider formulas as defined below:
\begin{equation}
        \begin{aligned}
             F, G &:= s \mid (F \wedge G) \mid (F \vee G) \mid (\neg F) \mid (F \to G),\\
             p &:= \sum^n_{i=1} w_i x_i \leq b \mid \sum^n_{i=1} w_i x_i < b,
        \end{aligned}    
\end{equation}
such that $F$ and $G$ are quantifier-free first-order formulas over the theory of linear real arithmetic. Moreover, $p$ represents the atomic formulas such that $n \geq 1$, each $w_i$ and $b$ are fixed real numbers, and each $x_i$ is a first-order variable. Observe that we allow the use of other letters for variables instead of $x_i$, such as $s_i$, $z_i$, $q_i$. For example, $(2.5x_1 + 3.1x_2 \geq 6) \wedge (x_1=1 \vee x_1=2) \wedge (x_1=2 \to x_2 \leq 1.1)$ is a formula by this definition. Observe that we allow standard abbreviations as $\neg (2.5x_1 + 3.1x_2 < 6)$ for $2.5x_1 + 3.1x_2 \geq 6$.

Since we are assuming the semantics of formulas over the domain of real numbers, an \textit{assignment} $\mathcal{A}$ for a formula $F$ is a mapping from the first-order variables of $F$ to elements in the domain of real numbers. For instance, $\{x_1 \mapsto 2.3, x_2 \mapsto 1\}$ is an assignment for $(2.5x_1 + 3.1x_2 \geq 6) \wedge (x_1=1 \vee x_1=2) \wedge (x_1=2 \to x_2 \leq 1.1)$. An assignment $\mathcal{A}$ \textit{satisfies} a formula $F$ if $F$ is true under this assignment. For example, $\{x_1 \mapsto 2, x_2 \mapsto 1.05\}$ satisfies the formula in the above example, whereas $\{x_1 \mapsto 2.3, x_2 \mapsto 1\}$ does not satisfy it.

A formula $F$ is \textit{satisfiable} if there exists a satisfying assignment of $F$. To give an example, the formula in the above example is satisfiable since $\{x_1 \mapsto 2, x_2 \mapsto 1.05\}$ satisfies it. As another example, the formula $(x_1 \geq 2) \wedge (x_1 < 1)$ is unsatisfiable since no assignment satisfies it. Given formulas $F$ and $G$, the notation $F \models G$ is used to denote \textit{logical consequence} or \textit{entailment}, i.e., each assignment that satisfies $F$ also satisfies $G$. As an illustrative example, let $F = (x_1 = 2 \wedge x_2 \geq 1)$ and $G = (2.5x_1 + x_2 \geq 5) \wedge (x_1=1 \vee x_1=2)$. Then, $F \models G$. The essence of entailment lies in ensuring the correctness of the conclusion $G$ based on the given premise $F$. In the context of computing explanations, as presented in \cite{Ignatiev_abduction}, logical consequence serves as a fundamental tool for guaranteeing the correctness of predictions made by ANNs. Therefore, our adaptation of the encoding proposed by \citet{tjeng} also incorporates the principles of entailment for computing explanations.

The relationship between satisfiability and entailment is a fundamental aspect of logic. It is widely known that, for all formulas $F$ and $G$, it holds that $F \models G$ iff $F \wedge \neg G$ is unsatisfiable. For instance, $( x_1 = 2 \wedge x_2 \geq 1) \wedge \neg((2.5x_1 + x_2 \geq 5) \wedge (x_1=1 \vee x_1=2))$ has no satisfying assignment since an assignment that satisfies $(x_1 = 2 \wedge x_2 \geq 1)$ also satisfies $(2.5x_1 + x_2 \geq 5) \wedge (x_1=1 \vee x_1=2)$ and, therefore, does not satisfy $\neg((2.5x_1 + x_2 \geq 5) \wedge (x_1=1 \vee x_1=2))$. Since our approach builds upon the concept of logical consequence, we can leverage this connection in the context of computing explanations for ANNs.

\subsection{Mixed Integer Linear Programming}
In Mixed Integer Linear Programming (MILP), the objective is to optimize a linear function subject to linear constraints, where some or all of the variables are required to be integers \cite{milp}. MILP is a crucial technique in our work for determining the lower and upper bounds of each neuron in the ANNs. For example, we utilize a minimization problem to determine the lower bound of neurons within ANNs. This process involves formulating an objective function that seeks to minimize the lower bound, subject to constraints that reflect the behaviour of ANNs. To illustrate the structure of a MILP, we provide an example below:
\begin{equation}
\label{eq:milp}
\begin{aligned}
\min \quad & y_1 \\
\textrm{s.t.} \quad & 1 \leq x_1 \leq 3\\
                    %&0 \leq x_2 \leq 1\\
                    %& 3x_1 + x_2 + s_1 - 2 = y_1 \\
                    & 3x_1 + s_1 - 2 = y_1 \\
                    %& 0 \leq y_1 \leq 3x_1 + x_2 - 2 \\
                    & 0 \leq y_1 \leq 3x_1 - 2 \\
                    %& 0 \leq s_1 \leq 3x_1 + x_2 - 2\\
                    & 0 \leq s_1 \leq 3x_1 - 2\\
                    & z_1 = 1 \to y_1 \leq 0 \\
                    & z_1 = 0 \to s_1 \leq 0 \\
                    & z_1 \in \{0, 1\}
%\label{eq:milp_example4} 
\end{aligned}
\end{equation}

In the MILP in (\ref{eq:milp}), we want to ﬁnd values for variables $x_1, y_1, s_1, z_1$ minimizing the value of the objective function $y_1$ among all values that satisfy the constraints. Variable $z_1$ is binary since $z_1 \in \{0, 1\}$ is a constraint in the MILP, while variables $x_1, y_1, s_1$ have the real numbers $\mathbb{R}$ as their domain. The constraints in a MILP may appear as linear equations, linear inequalities, and indicator constraints. Indicator constraints can be seen as logical implications of the form $z = v \to \sum^n_{i=1} w_i x_i \leq b$ such that $z$ is a binary variable, $v$ is a constant $0$ or $1$ \cite{bigm}.

An important observation is that a MILP problem without an objective function corresponds to a satisfiability problem, as discussed in Section~\ref{subsec:logic}. Given that the approach for computing explanations relies on logical consequence, and considering the connection between satisfiability and logical consequence, we employ a MILP solver to address explanation tasks. Additionally, throughout the construction of the MILP model, we utilize optimization, specifically employing a MILP solver, to determine tight lower and upper bounds for the neurons of ANNs.

\subsection{Classification Problems and Artificial Neural Networks}\label{ANN}

In machine learning, classification problems are defined over a set of $n$ features $\mathcal{F} = \{x_1, ..., x_n\}$ and a set of $\mathcal{N}$ classes $\mathcal{K} = \{c_1, c_2,...,c_\mathcal{N}\}$. In this work, we consider that each feature $x_i \in \mathcal{F}$ takes its values $v_i$ from the domain of real numbers. Moreover, each feature $x_i$ has an upper bound $u_i$ and a lower bound $l_i$ such that $l_i \leq x_i \leq u_i$, and its domain is the closed interval $[l_i, u_i]$. This is represented as a set of domain constraints or feature space $D = \{l_1 \leq x_1 \leq u_1,\text{ }l_2 \leq x_2 \leq u_2, ..., l_n \leq x_n \leq u_n \}$. For example, a feature for the height of a person belongs to the real numbers and may have lower and upper bounds of $0.5$ and $2.1$ meters, respectively. Furthermore, $\{x_1 = v_1, x_2 = v_2, ..., x_n = v_n\}$ represents a specific point or instance of the feature space such that each $v_i$ is in the domain of $x_i$.

An ANN is a function that maps elements in the feature space into the set of classes $\mathcal{K}$. A feedforward ANN is composed of $L+1$ layers of neurons. Each layer $l \in \{0, 1, ..., L\}$ is composed of $n_l$ neurons, numbered from $1$ to $n_l$. Layer $0$ is fictitious and corresponds to the input of the ANN, while the last layer, $K$ corresponds to its outputs. Layers $1$ to $L-1$ are typically referred to as hidden layers. Let $x^l_i$ be the output of the $i$th neuron of the $l$th layer, with $i \in \{1,...,n_l\}$. The inputs to the ANN can be represented as $x^0_i$ or simply $x_i$. Moreover, we represent the outputs as $x^L_i$ or simply $o_i$.

The values $x^l_i$ of the neurons in a given layer $l$ are computed through the output values $x^{l-1}_j$ of the previous layer, with $j \in \{1,...,n_{l-1}\}$. Each neuron applies a linear combination of the output of the neurons in the previous layer. Then, the neuron applies a nonlinear function, also known as an activation function. The output of the linear part is represented as $\sum_{j=1}^{n_{l-1}} w^{l}_{i,j} x^{l-1}_{j} + b^{l}_{i}$ where $w^{l}_{i,j}$ and $b^{l}_{i}$ denote the weights and bias, respectively, serving as parameters of the $i$th neuron of layer $l$. In this work, we consider only feedforward ANNs with the Rectified Linear Unit ($\mathrm{ReLU}$) as activation function because it can be represented by linear constraints due to its piecewise-linear nature. This function is a widely used activation whose output is the maximum between its input value and zero. Then, $x^{l}_{i} = \mathrm{ReLU}(\sum_{j=1}^{n_{l-1}} w^{l}_{i,j} x^{l-1}_{j} + b^{l}_{i})$ is the output of the $\mathrm{ReLU}$.

For classification tasks, the last layer $L$ is composed of $n_L = \mathcal{N}$ neurons, one for each class. Moreover, it is common to normalize the output layer using a Softmax layer. Consequently, these values represent the probabilities associated with each class. The class with the highest probability is chosen as the predicted class. However, we do not need to consider this normalization transformation as it does not change the maximum value of the last layer. Thus, the predicted class is $c_i \in \mathcal{K}$ such that $i = \argmax_{j \in \{1, ..., \mathcal{N}\}} x^L_j$.

\section{Explanations for ANNs with Logical Implications}\label{explanations}

%based on logic with guarantees
%on the correctness and minimality of explanations computes minimal explanations for ANNs, yielding a subset of the input features sufficient for the prediction.
%The algorithm proposed by \citet{Ignatiev_abduction} computes minimal explanations for ANNs, giving us a subset of the input features sufficient for the prediction. This approach is based on logic with guarantees on the correctness and minimality of explanations.

\citet{Ignatiev_abduction} proposed an algorithm that computes minimal explanations for ANNs, yielding a subset of the input features sufficient for the prediction. This approach is based on logic with guarantees on the correctness and minimality of explanations. A flowchart for computing explanations using such an algorithm is shown in Figure \ref{fig:fluxograma}.

\begin{figure}
\centering
\includegraphics[scale=0.6]{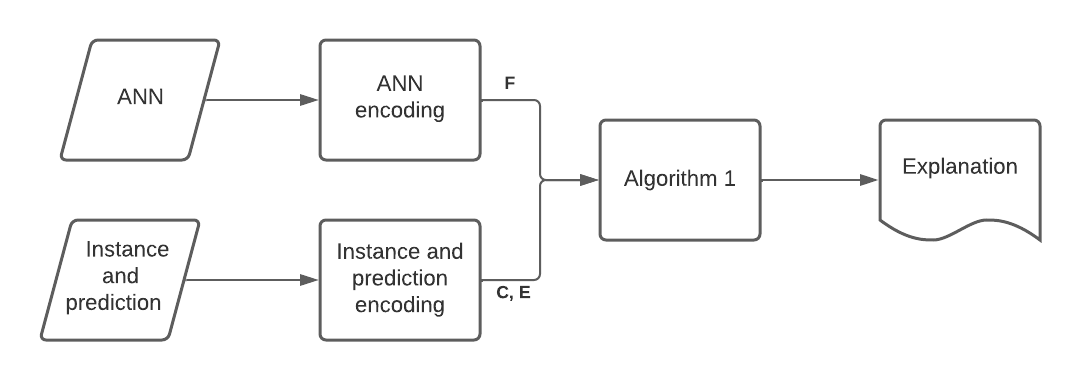}%[width=\textwidth]
\caption{Flowchart for calculating explanations.}\label{fig:fluxograma}
\end{figure}
%%oracle that solves MILP problems.
First, the ANN and the feature space $\{l_1 \leq x_1 \leq u_1,\text{ }l_2 \leq x_2 \leq u_2, ..., l_n \leq x_n \leq u_n \}$ are encoded as a formula $F$, an instance $\{x_1 = v_1, x_2 = v_2, ..., x_n = v_n\}$ of the feature space is encoded as a conjunction in a formula $C$, and the associated prediction by the ANN is encoded as a formula $E$. Then, it holds that $C \wedge F \models E$. The minimal explanation $C_m$ of $C$ is calculated removing feature by feature from $C$. For example, given a feature $x_i$ with value $v$ in $C$, if $C \setminus \{x_i=v\} \wedge F \models E$, feature $x_i$ may be considered as irrelevant in the explanation and is removed from $C$. Otherwise, if $C \setminus \{x_i=v\} \wedge F \not\models E$, then $x_i$ is kept in $C$ since the same class cannot be guaranteed. This $C \setminus \{x_i=v\}$ notation represents the removal of $x_i=v$ from formula $C$. This process is described in Algorithm \ref{algorithm1} and is performed for all features. Then, $C_m$ is the result at the end of this procedure. This means that for the values of the features in $C_m$, the ANN makes the same classification, whatever the values of the remaining features. Since to check entailments $C \wedge F \models E$ is equivalent to test whether $C \wedge F \wedge \neg E$ is unsatisfiable and $F$, $C$ and $\neg E$ are enconded as linear constraints and indicator constraints, such a entailment can be addressed by a MILP solver.
\begin{algorithm}
\caption{Computing a minimal explanation} \label{algorithm1}
\begin{algorithmic}[1]
\Require{ANN and domain constraints $F$, input data $C$, prediction $E$}
\Ensure{minimal explanation $C_m$}

\For{$x_i=v$ \text{ in } $C$}
\If{$C \setminus \{x_i=v\} \wedge F \models E$}
\State $C \gets C \backslash \{x_i=v\}$\
\EndIf
\EndFor
\State $C_m \gets C$\;
\State \Return $C_m$\;
%   %}
\end{algorithmic}
\end{algorithm}

The encoding of ANNs used in \cite{Ignatiev_abduction} and originally proposed by \citet{fischetti} uses implications to represent the behavior of the $\mathrm{ReLU}$ activation function. We encode an ANN with $L+1$ layers as in Equations (\ref{eq:indicator1})-(\ref{eq:indicator2}). In the following, we explain the notation. The encoding uses variables $x^{l}_{i}$ and $o_n$ with the same meaning as in the notation for ANNs. Auxiliary variables $s^{l}_{i}$ and $z^{l}_{i}$ control the behaviour of $\mathrm{ReLU}$ activations. Variable $z^{l}_{i}$ is binary and if $z^{l}_{i}$ is equal to $1$, the $\mathrm{ReLU}$ output $x^{l}_{i}$ is $0$ and $- s^{l}_{i}$ is equal to the linear part. Otherwise, the output $x^{l}_{i}$ is equal to the linear part and $s^{l}_{i}$ is equal to $0$. The constant $ub^{l}_{s,i}$ is the upper bound of variable $s^{l}_{i}$, and the constant $ub^{l}_{x,i}$ is the upper bound of variable $x^{l}_{i}$. Each variable $x^{0}_{i}$ has also lower and upper bounds $l_i$, $u_i$, respectively, defined by the domain of the features. 
\newcommand{\pd}[2]{\frac{\partial{#1}}{\partial{#2}}}
\begin{eqnarray}
     \left.
         \begin{aligned}
         &\sum_{j=1}^{n_{l-1}} w^{l}_{i,j} x^{l-1}_i + b^{l}_{i} = x^{l}_{i} - s^{l}_{i} \label{eq:indicator1}\\
         &z^{l}_{i} = 1 \rightarrow x^{l}_{n} \leq 0 \\
         &z^{l}_{i} = 0 \rightarrow s^{l}_{n} \leq 0 \\
         &z^{l}_{i} \in \{0, 1\} \\
         &0 \leq x^{l}_{i} \leq ub^{l}_{x,i} \\
         &0 \leq s^{l}_{i} \leq ub^{l}_{s,i} \\
         \end{aligned}
     \right\}l = 1, ..., L-1, \ i = 1, ..., n_l
\end{eqnarray}
\begin{align}
     &l_i \leq x_{i} \leq u_i,\quad i = 1, ..., n_0 \label{eq:input}\\
     &o_i = \sum_{j=1}^{n_{L-1}} w^{L}_{i,j} x^{L-1}_i + b^{L}_{i},\quad i = 1, ..., n_L \label{eq:indicator2}
\end{align}

The constraints in (\ref{eq:indicator1})-(\ref{eq:indicator2}) represent the formula $F$. The bounds $ub^{l}_{x,i}$ are defined by isolating variable $x^{l}_{i}$ from other constraints in subsequent layers. Then, $x^{l}_{i}$ is maximized to find its upper bound. A similar process is applied to find the bounds $ub^{l}_{s,i}$ for variables $s^{l}_{i}$. This optimization is possible due to the bounds of the features. Furthermore, these bounds can assist the solver in accelerating the computation of explanations. Therefore, the time required for this process must be considered when building $F$.

To check the unsatisfiability of the expression $C \wedge F \wedge \neg E$, we still need to take into account the formula $\neg E$, referring to the prediction of the ANN. Given an input $C$ predicted as class $c_i$ by the ANN, formula $E$ must be equivalent to $\bigwedge_{j=1, j \neq i}^{\mathcal{N}} o_i > o_j$. This formula asserts that the maximum value of the last layer is in output $o_i$. Therefore, $\neg E$ must ensure that $\bigvee_{j=1, j \neq i}^{\mathcal{N}} o_i \leq o_j$. Since MILP solvers can not directly represent disjunctions, we use implications (\ref{eq:inputindicator2}) and a linear constraint (\ref{eq:inputindicator3}) over binary variables to define $\neg E$.
\begin{align}
    %&x_{n} = v_{n}, \quad n = 1, ..., n_0 \label{eq:inputindicator1}\\ 
    &q_j = 1 \rightarrow o_i \leq o_j, \quad j \in \{1, ..., \mathcal{N}\} \setminus \{i\} \label{eq:inputindicator2} \\
    &\sum_{j=1, j \neq i}^{\mathcal{N}} q_j \geq 1 \label{eq:inputindicator3}\\
    &q_{j} \in \{0, 1\}, \quad j \in \{1, ..., \mathcal{N}\} \setminus \{i\} \label{eq:inputindicator4}
\end{align}

If an assignment $\mathcal{A}$ satisfies $\bigvee_{j=1, j \neq i}^{\mathcal{N}} o_i \leq o_j$, then $o_i \leq o_j$ is true under $\mathcal{A}$ for some $j$. Therefore, the assignment $\mathcal{A} \cup \{q_j \mapsto 1\}$ satisfies the Equations (\ref{eq:inputindicator2})-(\ref{eq:inputindicator4}). Conversely, if an assignment $\mathcal{A}$ satisfies Equations (\ref{eq:inputindicator2})-(\ref{eq:inputindicator4}), it clearly also satisfies $\bigvee_{j=1, j \neq i}^{\mathcal{N}} o_i \leq o_j$. We conclude this section with a proposition regarding the number of variables and constraints of the encoding discussed above.

\begin{proposition}
Let $C$ be an instance predicted as a class $c \in \mathcal{K}$ by an ANN, the formula $C \wedge F \wedge \neg E$ has $n_0 + n_L + 2\sum_{l=1}^{n_{L-1}} n_l$ real variables and $n_L - 1 + \sum_{l=1}^{n_{L-1}} n_l$ binary variables. Also, the formula has $n_0 + 2n_L + 5\sum_{l=1}^{n_{L-1}}$ constraints.

\end{proposition}

\section{Explanations for ANNs without Implications}\label{encoding}

%. Also, the approach in \cite{tjeng} encodes ANNs without . 

In this section, we present an adaptation of the encoding proposed by \citet{tjeng} for logic-based explainability. In such a work, the authors originally used the encoding to find adversarial examples without using logical implications. Even more importantly, such an encoding uses fewer variables and constraints compared to \cite{fischetti}. Then, we expect that our adaptation can lead to a better execution time for both building the logical constraints and computing explanations. Adapting the encoding in \cite{tjeng} to the context of computing explanations requires incorporating additional constraints that were not part of the original work. These new constraints represent the class predicted by the ANN as a formula $E$, as seen in Section~\ref{explanations}. However, to maintain the concept of the original encoding without implications, we define these additional constraints without implications.
%in \cite{tjeng} 

In the following, we apply the same algorithm from Figure \ref{fig:fluxograma}, but replacing the encoding of $F$ as in \cite{fischetti} with the one in \cite{tjeng}. We encode an ANN with $L+1$ layers as in Equations (\ref{eq:withoutindicator1}), (\ref{eq:input}) and (\ref{eq:indicator2}). The variables $x^{l}_{i}$ and $o_i$ have the same meaning as in Equations (\ref{eq:indicator1})-(\ref{eq:indicator2}). Furthermore, auxiliary variables $s^l_i$ are not required, as observed in the encoding by \citet{fischetti}. Constants $lb^{l}_{i}$ and $ub^{l}_{i}$ are, respectively, the lower and upper bounds of $\sum_{j=1}^{n_{l-1}} w^{l}_{i,j} x^{l-1}_j + b^{l}_{i}$. Again, we find such bounds via a MILP solver. The behavior of $\mathrm{ReLU}$ is modeled using these bounds and binary variables $z^{l}_{i}$. If $z^{l}_{i}$ is equal to $0$, the $\mathrm{ReLU}$ output $x^{l}_{i}$ is $0$. Otherwise, $x^{l}_{i}$ is equal to $\sum_{j=1}^{n_{l-1}} w^{l}_{i,j} x^{l-1}_j + b^{l}_{i}$. The bounds $lb^{l}{i}$ and $ub^{l}{i}$ are necessary to maintain the integrity of the set of constraints for the entire feature space. Regardless of the value of $z^{l}{i}$, the bounds ensure that the constraints remain valid for the entire feature space.
\begin{eqnarray}
     \left.
         \begin{aligned}
         &x^{l}_{i} \leq \sum_{j=1}^{n_{l-1}} w^{l}_{i,j} x^{l-1}_j + b^{l}_{i} - lb^{l}_{i} (1 - z^{l}_{i}) \label{eq:withoutindicator1}\\
         &x^{l}_{i} \geq \sum_{j=1}^{n_{l-1}} w^{l}_{i,j} x^{l-1}_j + b^{l}_{i} \\
         &x^{l}_{i} \leq ub^{l}_{i} z^{l}_{i} \\
         &z^{l}_{i} \in \{0, 1\} \\
         &x^{l}_{i} \geq 0 \\
         \end{aligned}
     \right\}l = 1, ..., L-1, \ i = 1, ..., n_l
\end{eqnarray}

%In the adaptation of the encoding in \cite{tjeng} for 
%the encoding in \cite{tjeng} to

%These elements are not originally present in the approach proposed by \citet{tjeng}. However, they are necessary for the context of computing explanations for ANNs. 
%We recall such elements are not originally present in
%Tjeng et al. [16]. 

In our proposal for computing explanations, constraints in Equations (\ref{eq:withoutindicator1}), (\ref{eq:input}) and (\ref{eq:indicator2}) represent the formula $F$. As in Section~\ref{explanations}, an instance is a conjunction $C$, and the associated prediction by the ANN is a formula $E$. Given an input $C$ predicted as class $c_i$ by the ANN, again formula $\neg E$ must ensure that $\bigvee_{j=1, j \neq i}^{\mathcal{N}} o_i \leq o_j$. Therefore, we must add new constraints to represent $\neg E$. Maintaining the concept of the original encoding in \cite{tjeng} without implications, we define these additional constraints accordingly. We employ binary variables $q_j$ and the upper and lower bounds $ub_j$ and $lb_j$ of variables $o_j$. As for $lb^{l}_{i}$ and $ub^{l}_{i}$, we find the bounds $ub_j$ and $lb_j$ through a MILP solver. We recall such elements are not originally present in \cite{tjeng}. However, they are necessary for
the context of computing explanations for ANNs. In Equations (\ref{eq:inputwithoutindicator1})-(\ref{eq:inputwithoutindicator2}) we represent our proposal for encoding formula $\neg E$, where the prediction associated with an input $C$ is class $c_i$. 
\begin{align}
    &o_i - o_j \leq (ub_i - lb_j) (1 - q_j) ,  \quad j \in \{1, ..., \mathcal{N}\} \setminus \{i\} \label{eq:inputwithoutindicator1} \\
    &\sum_{j=1, j \neq i}^{\mathcal{N}} q_j \geq 1 \label{eq:sum_at_least_1} \\
    &q_{j} \in \{0, 1\}, \quad j \in \{1, ..., \mathcal{N}\} \setminus \{i\} \label{eq:inputwithoutindicator2}
\end{align}

In what follows, we prove that Equations (\ref{eq:inputwithoutindicator1})-(\ref{eq:inputwithoutindicator2}) correctly ensure that $\bigvee_{j=1, j \neq i}^{\mathcal{N}} o_i \leq o_j$.

\begin{proposition}
Let $\neg E$ be defined as in Equations \emph{(\ref{eq:inputwithoutindicator1})-(\ref{eq:inputwithoutindicator2})}. Let $i \in \{1, ..., \mathcal{N}\}$ be fixed and $ub_i$ be such that $o_i \leq ub_i$. Let $lb_j$ be such that $lb_j \leq o_j$, for $j \in \{1, ..., \mathcal{N}\} \setminus \{i\}$. Therefore,
\[
\neg E \text{ is satisfiable iff } \bigvee_{j=1, j \neq i}^{\mathcal{N}} o_i \leq o_j \text{ is satisfiable.}
\]
\end{proposition}

\begin{proof}
If an assignment $\mathcal{A}$ satisfies $\bigvee_{j=1, j \neq i}^{\mathcal{N}} o_i \leq o_j$, then $o_i \leq o_j'$ is true under $\mathcal{A}$ for $j' \neq i$. Let $\mathcal{A}' = \mathcal{A} \cup \{q_{j} \mapsto v \mid v=1 \text{ if } j=j', \text{ else } v=0, \text{ for } j \in\{1, ..., \mathcal{N}\} \setminus \{i\}\}$ be an assignment. Then, $\mathcal{A}'$ imposes that $o_i \leq o_j'$ in Equation~(\ref{eq:inputwithoutindicator1}) for $j = j'$, which in clearly true under this assignment. For $j \neq j'$, it follows that $o_i - o_j \leq (ub_i - lb_j)$ must hold, which is also true under $\mathcal{A}'$ since $lb_j \leq o_j$ and $o_i \leq ub_i$.

Conversely, if an assignment $\mathcal{A}$ satisfies Equations (\ref{eq:inputwithoutindicator1})-(\ref{eq:inputwithoutindicator2}), it satisfies some $q_j'$, for $j' \neq i$ by Equation~(\ref{eq:sum_at_least_1}). Moreover, $\mathcal{A}$ satisfies $o_i \leq o_{j'}$ by Equation~(\ref{eq:inputwithoutindicator1}), for $j = j'$. 
Therefore, $\mathcal{A}$ also satisfies $\bigvee_{j=1, j \neq i}^{\mathcal{N}} o_i \leq o_j$.
\end{proof}

Finally, we give a proposition on the number of variables and constraints in our adaptation of the encoding in \cite{tjeng}. 

\begin{proposition}
Let $C$ be an instance predicted as a class $c \in \mathcal{K}$ by an ANN, then formula $C \wedge F \wedge \neg E$ has $n_0 + n_L + \sum_{l=1}^{n_{L-1}} n_l$ real variables and $n_L - 1 + \sum_{l=1}^{n_{L-1}} n_l$ binary variables. Moreover, the formula has $n_0 + 2n_L + 4\sum_{l=1}^{n_{L-1}} n_l$ constraints.
\end{proposition}

% in \cite{tjeng}
Therefore, this encoding has $\sum_{l=1}^{n_{L-1}} n_l$ fewer variables than the one presented in Section~\ref{explanations}. Additionally, this encoding has $\sum_{l=1}^{n_{L-1}} n_l$ fewer constraints. Consequently, one would expect a reduction in running time for both building the logical constraints and computing explanations.% 

\section{Experiments}\label{experiments}

%adaptation of the encoding proposed in \cite{tjeng}
%with
%experimental 
%To compare the times for computing explanations, we explained all instances in a given dataset and calculated the average time and standard deviation. 

In this section, we detail the experiments conducted to compare our proposal against the encoding presented in \cite{fischetti}. Our evaluation consists of two main experiments. In the first one, we compare the two encodings using $12$ datasets. In the second one, we conduct a detailed comparison using a single dataset. We vary the architecture of the trained ANNs to explore the effect of the number of layers and neurons. We evaluate the performance of each encoding in terms of time for building logical constraints and time for computing explanations. We explained all instances in a given dataset and calculated the average time and standard deviation to compare times for computing explanations. In the other case, given a trained ANN on a dataset, we built the logical constraints $10$ times and calculated the average time and standard deviation. 

%the used the considered  their for training
%obtained
%each of the considered metrics: 
%We also highlight the 

Next, we present the experimental setup, describing technologies, datasets, the trained ANNs and hyperparameters. After that, we discuss the results providing a comparative analysis of the running times for building logical constraints and computing explanations. Finally, we highlight specific improvements observed in our proposal.

%with the adaptation of the encoding proposed by \citet{tjeng}. 

\subsection{Experimental Setup}

We used Python to implement the approaches and to run the experiments. TensorFlow was used to manipulate ANNs, including the training and testing steps. CPLEX was used as the MILP solver and accessed by the DOcplex library.

We used 12 datasets from the UCI Machine Learning Repository\footnote{\url{https://archive.ics.uci.edu/ml/}} and Penn Machine Learning Benchmarks\footnote{\url{https://github.com/EpistasisLab/penn-ml-benchmarks/}}, each ranging from 9 to 32 integer, continuous, categorical or binary features. The number of instances in the selected
datasets ranges from 156 to 691. The types of classification problems related to these datasets are binary and multi-class classification. The preprocessing performed on the datasets included one-hot encoding of the categorical data and normalization of the continuous features to the range $[0, 1]$. This normalization was not applied to the integer features to avoid transforming their space into continuous, which could compromise formal guarantees on the correctness of the algorithm. As far as we know, such a methodology was not considered in earlier works.

The ANNs training was accomplished using a batch size of $4$ and a maximum of $100$ epochs, applying early stopping regularization with $10$ epochs based on validation loss. The optimization algorithm used was Adam and the learning rate was $0.001$. The datasets split was $80\%$ for training and $20\%$ for validation. The ANN architectures were limited to $2$ layers to reduce the total running time, because many solver calls were performed in the experiments due to the large number of instances. Each solver call deals with an NP-complete problem, therefore, impacting the experiments running time.

%performed is the comparison of

The first experiment compared the two encodings presented on the 12 datasets. The architecture of the trained ANNs is two hidden layers with $20$ neurons each. For each dataset and the associated ANN, the explanation of each instance was obtained using the Algorithm \ref{algorithm1} and both presented encodings. The second experiment performed the comparison of the two encodings presented using the voting dataset of the first experiment. This experiment was conducted in two cases. In the first case, the trained ANNs consist of one hidden layer with the number of neurons ranging from $10$ to $100$. In the second case, the ANNs consist of two hidden layers such that both layers contains the same number of neurons, ranging from $10$ to $40$. In both cases, the number of neurons in the layers increases in increments of $5$. Again, the explanations of each instance was obtained using Algorithm \ref{algorithm1} and both presented encodings. The objective of this experiment is to verify the influence of the number of layers and the number of neurons in both encodings.

\subsection{Results}
%, \citet{tjeng}

\begin{table}[t]
    \centering
    \setlength{\tabcolsep}{8pt}
    \resizebox{!}{2.7cm}{
    % \begin{tabular}{|c|c|c|c|c|c|c|}
    \begin{tabular}{cccccc}
    \hline
    \multirow{2}{2cm}{\textbf{Datasets}} &
    %\multirow{2}{2cm}{\textbf{Explanation Size}} & 
    \multicolumn{2}{c}{\textbf{\citet{fischetti}}} & \multicolumn{2}{c}{\textbf{Our Proposal}}\\
    % \hline
    % \textbf{Inactive Modes} & \textbf{Description}\\
    \cline{2-5}
    &\textbf{Exp (s)} & \textbf{Build (s)} & \textbf{Exp (s)} & \textbf{Build (s)}\\
    %\hhline{~--}
    \hline
    
    breast-cancer (9) & 0.39 \textpm 0.24 & 2.98 \textpm 0.05 & 0.39 \textpm 0.21 & \textbf{2.67 \textpm 0.17} \\ 
    \hline
    
    glass (9) & 0.25 \textpm 0.09 & 4.25 \textpm 0.18 & 0.26 \textpm 0.1 & \textbf{3.86 \textpm 0.09} \\ 
    \hline
    
    glass2 (9) & 0.39 \textpm 0.29 & 3.45 \textpm 0.1 & 0.47 \textpm 0.34 & \textbf{3.02 \textpm 0.02} \\
    \hline
    
    cleve (13) & 0.51 \textpm 0.23 & 4.26 \textpm 0.14 & 0.54 \textpm 0.25 & \textbf{3.55 \textpm 0.01}\\
    \hline
    
    cleveland (13) & 0.69 \textpm 0.8 & 5.28 \textpm 0.11 & 0.82 \textpm 1.03 & \textbf{4.57 \textpm 0.17} \\
    \hline
    
    heart-statlog (13) & 0.41 \textpm 0.18 & 3.76 \textpm 0.05 & 0.41 \textpm 0.18 & \textbf{3.06 \textpm 0.02}\\
    \hline
    
    australian (14) & 0.76 \textpm 0.43 & 3.2 \textpm 0.08 & 0.84 \textpm 0.42 & 3.14 \textpm 0.72\\
    \hline
    
    voting (16) & 1.02 \textpm 0.52 & 3.62 \textpm 0.07 & 1 \textpm 0.47 & \textbf{3.01 \textpm 0.02} \\
    \hline
    
    hepatitis (19) & 1.17 \textpm 1.27 & 5.66 \textpm 0.14 & 1.1 \textpm 1.13 & \textbf{5.28 \textpm 0.07} \\
    \hline
    
    spect (22) & 2.64 \textpm 1.56 & \textbf{4.81 \textpm 0.13} & 3.21 \textpm 1.82 & 5.74 \textpm 0.15 \\
    \hline
    
    auto (25) & 0.79 \textpm 0.3 & 6.86 \textpm 0.25 & 0.83 \textpm 0.32 & 6.62 \textpm 0.14 \\
    \hline
    
    backache (32) & 11.44 \textpm 10.3 & 5.04 \textpm 0.16 & 11.39 \textpm 9.74 & \textbf{4.78 \textpm 0.14}\\
    \hline
    
  \end{tabular}%
  }
    \caption{Comparison of both encodings.}
    \label{tab:results1}
\end{table}

The results of the first experiment are shown in Table \ref{tab:results1}. For each dataset, its number of features is indicated in parentheses. The column \textbf{Exp (s)} refers to the average running time for computing explanations in seconds, and the standard deviation is also presented. The column \textbf{Build (s)} refers to the average running time, in seconds, for building the logical constraints of the trained ANN. The running time for finding the bounds of variables is included in the time for building the encodings. %

%the encoding by \citet{tjeng}

Despite the encoding by \citet{fischetti} achieved a better average running time for computing explanations in $7$ out of $12$ cases, both encodings generally perform similarly when considering the variability of the time. For instance, in the spect dataset, the average execution time of the encoding by \citet{fischetti} ($2.64$ seconds) falls within the range of the average minus the standard deviation of our approach ($3.21 - 1.81$). This pattern is observed across several datasets. In summary, the results indicate that both encodings generally perform similarly in terms of running time for computing explanations, with minor variations across different datasets. 

%\vspace{-0.03cm}
%of the encoding proposed by \citet{tjeng}
%, where the encoding by \citet{fischetti} showed better performance
%of the encoding proposed by \citet{tjeng}
%the encoding proposed by \citet{tjeng} 
%the encoding proposed by \citet{tjeng}
%the encoding by \citet{tjeng}
%the encoding by \citet{tjeng}

With respect to the average running time for building the logical constraints, our adaptation generally outperformed the encoding by \citet{fischetti}. However, there were exceptions noted, such as in the spect dataset. Overall, our adaptation achieved an improvement of up to $18\%$ compared to the the other one, as seen in the heart-statlog dataset. In summary, the results indicate that our adaptation is consistently more efficient than the other approach for building logical constraints. The variability in the results, as shown by the standard deviations, further supports this conclusion. For instance, the average running time of our proposal is less than the average minus the standard deviation of the other approach in $9$ out of $12$ datasets. These cases are highlighted in bold in Table~\ref{tab:results1}. It is important to note that, in the spect dataset, the average running time of the encoding proposed by \citet{fischetti} is less than the average minus the standard deviation of our adaptation. This indicates that, while our proposal appears generally more efficient, specific dataset characteristics can influence which encoding is more advantageous. Our adaptation yielded notably superior results not only in the average time for building logical constraints but also in the overall time, which includes both computing explanations and constructing logical constraints. For instance, it achieved an improvement of up to $16\%$ compared to the other approach, as seen in the heart-statlog dataset. 

%of the encoding by \citet{tjeng}

\begin{figure}[t]
%\centering
\makebox[\textwidth][c]{\includegraphics[width=1.2\textwidth]{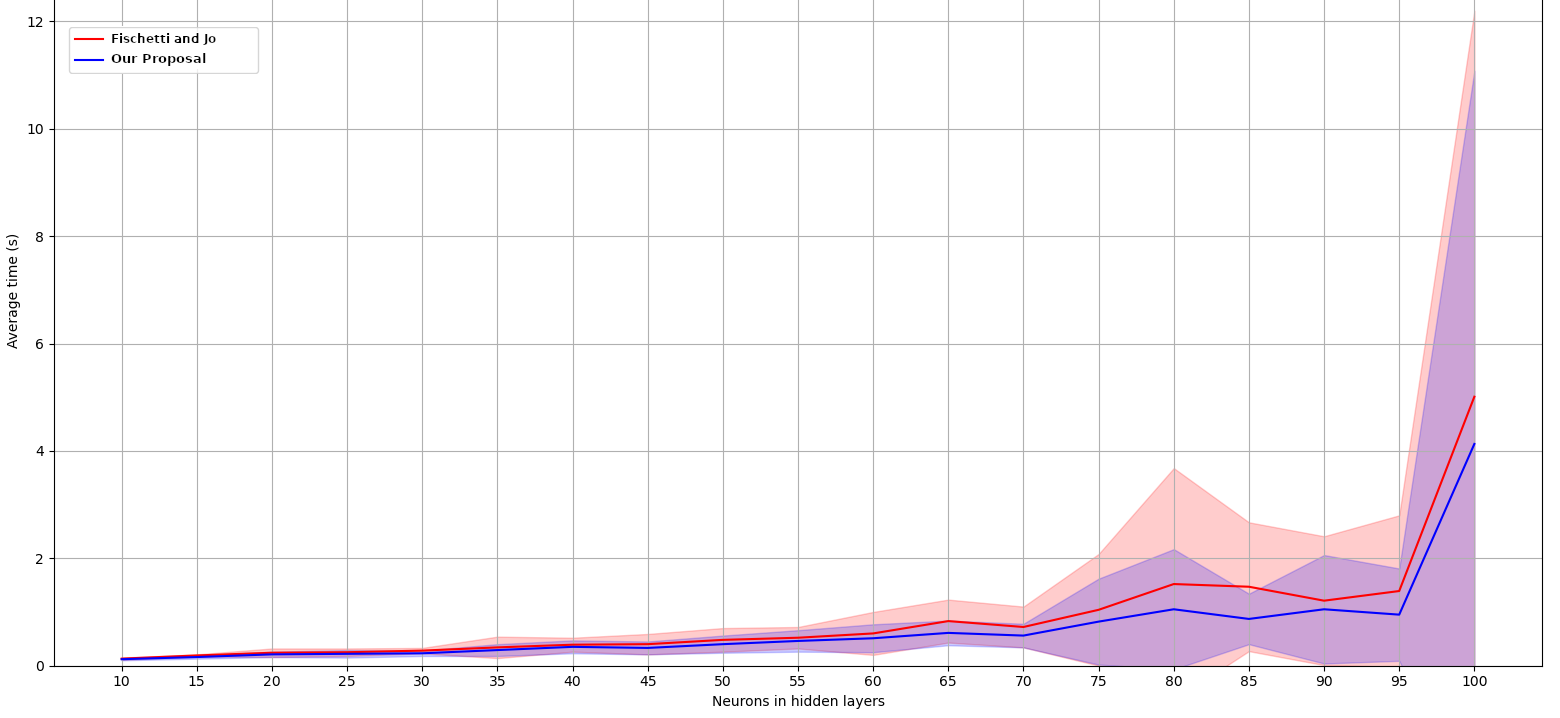}}%
\caption{Comparison of average running time for computing explanations, using ANNs with one hidden layer and the voting dataset.}\label{fig:time_voting}
\end{figure}

\begin{figure}[t]
\makebox[\textwidth][c]{\includegraphics[width=1.2\textwidth]{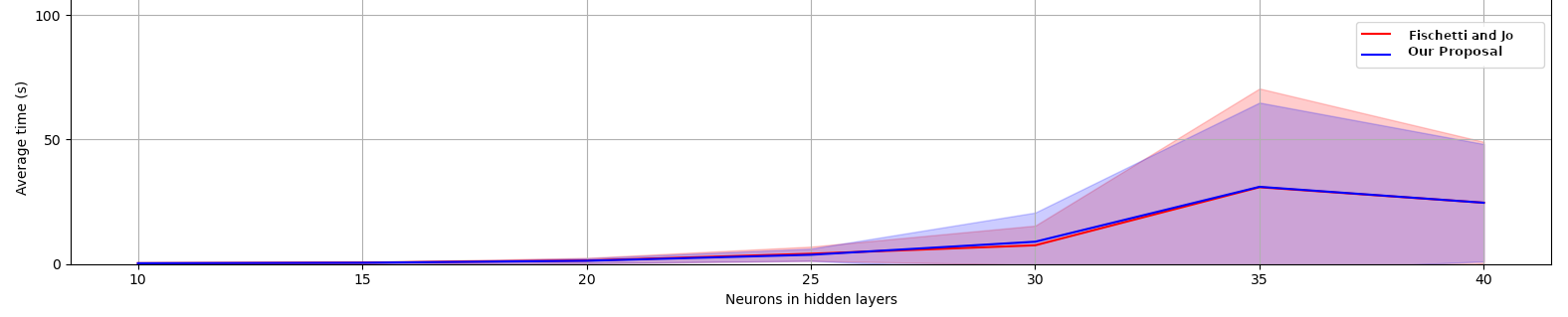}}%
\caption{Comparison of average running time for computing explanations, using ANNs with two hidden layers and the voting dataset.}\label{fig:time_voting_2layers}
\end{figure}

The results of the second experiment are shown in Figures \ref{fig:time_voting} and \ref{fig:time_voting_2layers}. The $x$-axis shows the number of neurons in each hidden layer in both Figures. While the $y$-axis shows the average running time for computing explanations. Furthermore, the standard deviation is indicated by the shaded region in both Figures. Figure \ref{fig:time_voting} refers to ANNs with one hidden layer, while Figure \ref{fig:time_voting_2layers} refers to ANNs with two hidden layers.

%the encoding proposed by \citet{tjeng}
%the encoding proposed by \citet{tjeng}
%the encoding by \citet{tjeng}

Figure \ref{fig:time_voting} suggests that our proposal, for ANNs with one hidden layer, achieved a superior average running time for computing explanations in the voting dataset. This encoding outperforms the other approach with percentage improvements ranging from approximately $7.69\%$ to $40.82\%$. The most significant improvements are observed with higher neuron counts. For instance, with $85$ neurons per layer, our proposal shows an improvement of around $40.82\%$, and with $95$ neurons, the improvement is about $31.65\%$. Moreover, our proposal generally exhibits similar or lower standard deviations, indicating more consistent performance across different number of neurons.

On the other hand, Figure \ref{fig:time_voting_2layers} depicts comparable results between both encodings for ANNs with two hidden layers. Moreover, similar standard deviations were achieved for both encodings, indicating comparable levels of consistency in performance. The variability increases significantly with more complex networks. For example, with $35$ and $40$ neurons per layer, both encodings exhibit significant variability. Furthermore, with $30$ neurons per layer, our adaptation shows higher variability compared to the other approach.
%the encoding by \citet{tjeng}
%

\section{Conclusions and future work}\label{conclusions}

Explanations for the outputs of ANNs are fundamental in many scenarios, due to critical systems, data protection laws, adversarial examples, among others. Therefore, several heuristic methods have been developed to provide explanations for the decisions made by ANNs. However, these approaches lack guarantees of correctness and may also produce redundant explanations. Logic-based approaches address these issues but often suffer from scalability problems.

%the encoding proposed by \citet{tjeng}
%adaptation of the encoding proposed by \citet{tjeng}

In this work, we compare two logical encodings of ANNs: 
one has been used in the literature to provide explanations \cite{fischetti, Ignatiev_abduction}, and another \cite{tjeng} that we have adapted for our context of explainability. Our experiments indicate that both encodings have similar running times for computing explanations, even as the number of neurons and layers increases. Furthermore, our experimental results suggest that our adaptation is generally more efficient for ANNs with one hidden layer, while the performance advantage diminishes for ANNs with two hidden layers. However, our proposal achieved a better running time for building the encoding for ANNs with two layers, showing an improvement of up to $18\%$. This can help to decrease the scalability issue for building the logical constraints given an ANN. Furthermore, this encoding obtained better results also in the overall time, i.e., the time for computing explanations plus the time for building the logical constraints, showing an improvement of up to $16\%$.

%, such as proposed by \cite{reluplex} and \cite{katz}.
% finding minimal size explanations

In the experiments of this work, we considered all instances of the datasets used, which considerably increased the experiments running time. As future work, we can change the design of experiments, using only a subset of the datasets, to allow the use of larger ANNs. More experiments are necessary, especially with additional layers and neurons, to further validate our findings and understand the performance of these encodings. Furthermore, others encodings \cite{reluplex,katz} can be evaluated for computing logic-based minimal explanations for ANNs. Moreover, in order to improve the scalability of computing logic-based explanations, the ANNs can be simplified, before or during building their encodings, via pruning or slicing as proposed by \cite{pruning}. This results in equivalent ANNs with smaller sizes.

\bibliographystyle{splncs04nat}%
\bibliography{references}

\end{document}